\documentclass{ap-jnmp}

\newcommand{\C}{\mathbb{C}}

\renewcommand{\L}{\mathcal{L}}

\renewcommand{\x}{\mathbf{x}}

\def\({\left(}
\def\){\right)}

\def\[{\begin{eqnarray}}
\def\]{\end{eqnarray}}
\def\d{\partial}

\def\d{\partial}

\copyrightauthor{}

\title{Integrability on generalized $q$-Toda equation and hierarchy}

\author{\footnotesize Anni Meng}

\address{Department of
Mathematics,  Ningbo University\\
Ningbo, 315211, Zhejiang, P.\ R.\ China}

\author{\footnotesize Chuanzhong Li\footnote{%
Corresponding author:lichuanzhong@nbu.edu.cn.} \ \ and  Shuo Huang}

\address{Department of
Mathematics,  Ningbo University\\
Ningbo, 315211, Zhejiang, P.\ R.\ China\\
\email{}}

\begin{document}

\maketitle
\thispagestyle{empty}

\vphantom{\vbox{%
\begin{history}
\received{(26 March 2014)}
\revised{(Day Month Year)}
\accepted{(12 May 2014)}
\end{history}
}}
\begin{abstract}
 In this paper, we construct a new integrable equation which is a generalization of $q$-Toda equation. Meanwhile its soliton solutions are constructed to show
 its integrable property. Further the Lax pairs of the generalized $q$-Toda equation and a whole integrable generalized $q$-Toda hierarchy are also constructed.
 To show the integrability, the Bi-hamiltonian structure and tau symmetry of the generalized $q$-Toda hierarchy are given and this leads to the tau function.
\end{abstract}


\ccode{2000 Mathematics Subject Classification: 37K05, 37K10, 37K20}

\section{Introduction}
The Toda lattice equation is a completely integrable system which has many important applications in mathematics and physics including the theory of Lie algebra representation, orthogonal polynomials and  random
matrix model  \cite{Toda,Todabook,UT,witten,dubrovin}. Toda system has many kinds of reduction or extension, for example extended Toda hierarchy (ETH)\cite{CDZ}, bigraded Toda hierarchy (BTH)\cite{C}-\cite{EQTH} and so on. These generalized Toda hierarchies have important application in Gromov-Witten theory on $\C P^1$ and orbiford.

The $q$-calculus ( also called quantum calculus)  traces
back to the early 20th century and attracted  important
works in the area of $q$-calculus\cite{Jackson,kac}  and $q$-hypergeometric series. The
$q$-deformation of  classical nonlinear integrable system  started in 1990's by means of
$q$-derivative $\partial_q $   instead of usual derivative with respect spatial variable  in the classical system.
Several $q$-deformed integrable systems have been presented, for example the
$q$-deformed Kadomtsev-Petviashvili ($q$-KP) hierarchy is a
subject of intensive study in the literatures \cite{mas}-\cite{myqckp}.
  The $q$-Toda equation was also studied in \cite{ZTAK,Silindir} but not for a whole hierarchy. This paper will be devoted to the further studies on a generalized  $q$-Toda equation(GQTE) and generalized  $q$-Toda hierarchy(GQTH).

To show the complete integrability of nonlinear evolution, it is necessary to test whether the equation has Hirota bilinear equation, three-soliton solution,  Lax pair, Bi-hamiltonian structure and even tau symmetry. This paper will show the integrability on the Generalized $q$-Toda hierarchy from the above several directions.

\section{$q$-difference operator and its generalization}
As we all know, in common sense an integrable equation can always be rewritten in form of a Hirota bilinear equation using Hirota direct method. Therefore firstly we introduce some basic notation including Hirota derivatives as a preparation for introducing the Hirota bilinear equation of the generalized  $q$-Toda equation.

Let $F$ be a space of differentiable functions $f,g: \mathbb{R}^{n}\rightarrow \mathbb{R}$.
The Hirota $D$-operator $D:F\times F\rightarrow F$ is defined as
\begin{equation}\label{equa3}
[D^{m_1}_{x}D^{m_2}_{t}...]f\cdot g=[(\partial_{x}-\partial_{x'})^{m_1}(\partial_{t}-\partial_{t'})^{m_2}...]f(x,t,...)g(x',t',...)|_{x'=x,t'=t,...}.
\end{equation}
Then one can find the following standard statement holds.
Let $P(D)$ be an arbitrary polynomial in D acting on two differentiable
functions $f(x,t,...)$ and $g(x,t,...)$,then the following equations hold
\begin{equation}\label{equa4}
\qquad P(D)f\cdot g=P(-D)g\cdot f,
\end{equation}
\begin{equation}\label{equa5}
\qquad P(D)f\cdot 1=P(\partial)f;\ \ P(D)1\cdot f=P(-\partial)f,
\end{equation}
where $\partial$ is the usual differential operator with respect to spatial variable $x$.
The virtue of exponential identity can appropriately be as following form in terms of the Hirota D-operator
\begin{equation}\label{equa6}
 e^{\epsilon D_{x}}f(x)g(x)=f(x+\epsilon)g(x-\epsilon).
\end{equation}
If $\epsilon$ is parameter and $f,g$ belong to continuously differentiable functions, like in \cite{Silindir}, then define
\begin{equation}\label{equa7}
\sigma_{\epsilon}(x)=e^{\epsilon\x(x)\partial_{x}}x.
\end{equation}
then
\begin{equation}\label{equa7'}
e^{\epsilon\x(x)\partial_{x}}u(x)=u(e^{\epsilon\x(x)\partial_{x}}x)=u(\sigma_{\epsilon}(x)),\ \ \epsilon>0.
\end{equation}
If $\sigma_{\epsilon}(u(x))=e^{\epsilon\partial_{x}}u(x)=u(x+\epsilon)$,
the system introduced later will lead to original Toda lattice.
If $\sigma_{\epsilon}(x)=e^{\epsilon x\partial_{x}}x=e^{\epsilon}x,$  which implies $e^{\epsilon x\partial_{x}}u(x)=u(e^{\epsilon}x).$
Then the system will lead to $q$-Toda lattice in \cite{Silindir}.
 Considering that the vector field of the form $\x(x)\partial_{x}=x^{n}\partial_{x}$ on $\mathbb{R}$, it will be the general generalized $q$-Toda lattice. In this paper, we only give the case $n=2$, and we just name the leading system later the generalized $q$-Toda equation.
\begin{proposition}\label{theo}
The $q$-exponential identity acts on arbitrary continuous differentiable functions $f(x),g(x)$ as the rule
\begin{equation}\label{equa9}
e^{\epsilon x^{2}D_{x}}f(x)g(x)=\Lambda_{\epsilon}f(x)\Lambda_{\epsilon}^{-1}g(x),x\in\mathbb{R}
\end{equation}
where the forward and backward shift operators are separately represented by $\Lambda_{\epsilon}$ and $\Lambda_{\epsilon}^{-1}$, respective acting as
\begin{equation}\label{equa10}
\Lambda_{\epsilon}f(x)=f(\frac{x}{1-x\epsilon}),\qquad \Lambda_{\epsilon}^{-1}g(x)=g(\frac{x}{1+x\epsilon}).
\end{equation}
\end{proposition}
\begin{proof}
Making use of the change of variable
 $x^{2}D_{x}$=$D_{x'}, \ x'=-\frac1x$ is the idea to prove the identity, i.e.
 \begin{equation}\label{equa11}
e^{\epsilon x^{2}D_{x}}f(x)g(x)=e^{\epsilon D_{x'}}f(-\frac{1}{x'})g(-\frac{1}{x'}).
 \end{equation}
 Using eq.\eqref{equa6} for the right hand side of eq.\eqref{equa11}, we end up the proof with\\
 \\
$e^{\epsilon x^{2}D_{x}}f(x)g(x)=f(-\frac{1}{x'+\epsilon})g(-\frac{1}{x'-\epsilon})=f(\frac{1}{\frac{1}{x}-\epsilon})g(\frac{1}{\frac{1}{x}+\epsilon})=
\Lambda_{\epsilon}f(x)\Lambda_{\epsilon}^{-1}g(x).$\\
 \end{proof}To give the definition of the generalized $q$-Toda equation, we need the following cental generalized difference operator.
\begin{definition} \label{deflax}
The central q-difference operator $\triangle_{x}^{2}$ acts on an arbitrary function $f(x),x\in\mathbb{R}$,as
\begin{equation}\label{equa12}
\triangle_{x}^{2}f(x)=f(\frac{x}{1-x\epsilon})+f(\frac{x}{1+x\epsilon})-2f(x).
\end{equation}
which is easily rewritten as
$\triangle_{x}^{2}f(x)=(\Lambda_{\epsilon}+\Lambda_{\epsilon}^{-1}-2)f(x)$.
\end{definition}

In the next section, we will try to use the above defined generalize $q$-shift operator to define the generalized $q$-Toda\ equation.

\section{The \ generalized $q$-Toda\ equation}
The well-known Toda equation eq.\eqref{equa17} represents the motion of the one-dimensional particles by
\begin{equation}\label{equa13}
 \frac{d^{2}y_{n}}{d t^{2}}=e^{y_{n-1}-y_{n}}-e^{y_{n}-y_{n+1}}.
\end{equation}
 By introducing the force
\begin{equation}\label{equa14}
U_{n}=e^{y_{n-1}-y_{n}}-1.
\end{equation}
 the Toda equation eq.\eqref{equa13} turns out to be
\begin{equation}\label{equa15}
\frac{d^{2}}{dt^{2}}\log (1+U_{n})=U_{n+1}+U_{n-1}-2U_{n}.
\end{equation}
Similarly as Toda equation, we define the generalized $q$-Toda equation(GQTE) as follows
\begin{equation}\label{equaGQT}
 \epsilon^2\frac{d^{2}\phi(x)}{d t^{2}}=e^{\phi(\frac{x}{1+x\epsilon})-\phi(x)}-e^{\phi(x)-\phi(\frac{x}{1-x\epsilon})},
\end{equation}
 By introducing the force
\begin{equation}\label{equaGQT}
V=e^{\phi(\frac{x}{1+x\epsilon})-\phi(x)}-1,
\end{equation}
then the GQTE becomes
\begin{equation}\label{equa16}
\epsilon^2\frac{d^{2}}{dt^{2}}\log(1+V(x,t))=\triangle_{x}^{2}V(x,t)=V(\frac{x}{1-x\epsilon},t)+V(\frac{x}{1+x\epsilon},t)-2V(x,t).
\end{equation}
It is necessary to introduce the dependent variable transformation as
\begin{equation}\label{equa17}
V(x,t)=\frac{d^{2}}{dt^{2}}\log f(x,t).
\end{equation}
Then the bilinear form for $f(x,t)$ is evolved as
\begin{equation}\label{equa18}
V(x,t)=\frac{f_{tt}f-f_{t}^{2}}{f^{2}}=\frac{f(\frac{x}{1-x\epsilon},t)f(\frac{x}{1+x\epsilon},t)}{f^{2}}-1.
\end{equation}
Then the generalized $q$-Toda equation can be rewritten as a Hirota bilinear form in terms of Hirota D-operator as
\begin{equation}\label{equa19}
P(D)f(x,t)\cdot f(x,t)=[D_{t}^{2}-(e^{\epsilon x^{2}D_{x}}+e^{-\epsilon x^{2}D_{x}})-2)]f(x,t)\cdot f(x,t)=0,
\end{equation}
by multiplying eq.\eqref{equa18} by $2f^{2}(x,t)$ and using the q-exponential identity eq.\eqref{equa9}. Supposing function $f$ has finite perturbation expansion around a formal perturbation parameter
$\varepsilon$  as
\begin{equation}\label{equa20}
f(x,t)=1+\varepsilon f^{(1)}(x,t)+\varepsilon^{2}f^{(2)}(x,t)+...
\end{equation}
Substituting  eq.\eqref{equa20} into generalized $q$-Toda equation
\begin{equation}\label{equa21}
P(D)f(x,t)\cdot f(x,t)=0,
\end{equation}
we have
\begin{equation}\label{equa22}
\begin{split}
&P(D)f(x,t)\cdot f(x,t) \\
&=P(D)[1\cdot 1+\varepsilon 1(\cdot f^{(1)}+f^{(1)}\cdot 1)+\varepsilon^2 (1\cdot f^{(2)}+f^{(2)}\cdot 1+f^{(1)}\cdot f^{(2)})\\
&+\varepsilon^{3}(1\cdot f^{(3)}+f^{(3)}\cdot 1+f^{(1)}\cdot f^{(2)}+f^{(2)}\cdot f^{(1)})\\
&+\varepsilon^{4}(1\cdot f^{(4)}+f^{(4)}\cdot 1+f^{(1)}\cdot f^{(3)}+f^{(3)}\cdot f^{(1)}+f^{(2)}\cdot f^{(2)})+...].
\end{split}
\end{equation}
The coefficient of the first term $\varepsilon^{0}$ is trivial. For the coefficient of $\varepsilon^{1}$, we get
\begin{equation}\label{equa23}
P(D)1\cdot f^{(1)}+f^{(1)}\cdot 1=2P(\partial)f^{(1)}=2[\partial_{t}^{2}-(e^{\epsilon x^{2}\partial_{x}}+e^{-\epsilon x^{2}\partial_{x}}-2)]f^{(1)}=0.
\end{equation}
Then the equation $f^{(1)}$ has exponential type solution as
\begin{equation}\label{equa24}
f^{(1)}(x,t)=e^{-\frac{\alpha}{x}+\beta t+\eta},
\end{equation}
where $\alpha,\beta,\eta $ are arbitrary constants with the dispersion relation as
\begin{equation}\label{equa25}
\beta^{2}=e^{\alpha \epsilon}+e^{-\alpha \epsilon}-2.
\end{equation}
Comparing the coefficients of $\varepsilon^{2}$ in eq.\eqref{equa22} will yield
\begin{equation}P(D)1\cdot f^{(2)}+f^{(2)}\cdot 1+f^{(1)}\cdot f^{(1)}=2P(\partial)f^{(2)}+P(D)f^{(1)}\cdot f^{(1)}=0,
\end{equation}
which implies exactly
\begin{equation}\label{equa26}
[D_{t}^{2}-(e^{\epsilon x^{2}D_{x}}+e^{-\epsilon x^{2}D_{x}})-2)]f^{(1)}(x,t)\cdot f^{(1)}(x,t)
\end{equation}
\begin{equation}\notag
=-2[\partial_{t}^{2}-(e^{\epsilon x^{2}\partial_{x}}+e^{-\epsilon x^{2}\partial_{x}}-2)]f^{(2)}(x,t).
\end{equation}
Since $f^{(1)}$ given in eq.\eqref{equa24} satisfies the form of eq.\eqref{equa26} by considering eq.\eqref{equa25}, it is logical to take all order terms as zero,  i.e.$f^{(j)}=0,j\geq 2$. Therefore without loss of generality, we let $\varepsilon=1$. Then one-q-soliton is constructed by the virtue of eq.\eqref{equa24} and eq.\eqref{equa25}
as
\begin{equation}\label{equa28}
V(x,t)=\frac{\beta^{2}e^{-\frac{\alpha}{x}+\beta t+\eta}}{(1+e^{-\frac{\alpha}{x}+\beta t+\eta})^{2}}.
\end{equation}
The solution of one-q-soliton $V$ can be seen from Figure \ref{1}.
\begin{figure}[h!]
\centering
\raisebox{0.85in}{}\includegraphics[scale=0.45]{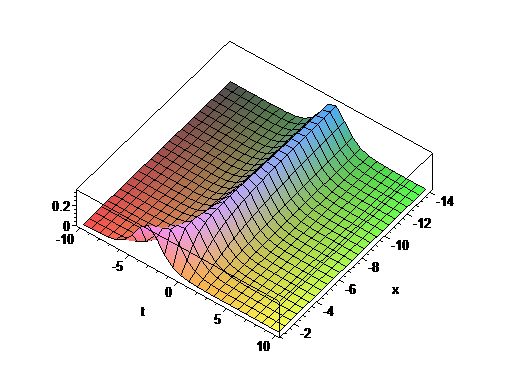}
\caption{One-q-soliton solution $V$ of generalized $q$-Toda equation with $e^{\epsilon}=1.25,\alpha_{1}=-5,\beta_{1}=-1.1745$.}\label{1}
\end{figure}

We pick the starting solution of  eq.\eqref{equa23}  as the assumption of two-soliton solutions.
\begin{equation}\label{equa29}
f^{(1)}=2\cosh(-\frac{\alpha_{1}}{x}+\beta_{1}t+\eta_{1}),
\end{equation}
where $\alpha_{i},\eta_{i},i=1,2$ are arbitrary constants with
the related dispersion relation
\begin{equation}\label{equa30}
\beta_{i}^{2}=e^{\alpha_{i}\epsilon}+e^{-\alpha_{i}\epsilon}-2, i=1,2.
\end{equation}
Apparently the use of vector notation
\begin{equation}\label{equa31}
p_{1}\pm p_{2}=(\beta_{1}\pm \beta_{2},\alpha_{1}\pm \alpha_{2},\eta_{1}\pm \eta_{2}),
\end{equation}
can lead to dispersion relation eq.\eqref{equa30} as $P(p_i)=0, i=1,2....$ Then we get
\begin{equation}\label{equa32}
-P(\partial)f^{(2)}=[(\beta_{1}-\beta_{2})^{2}-(e^{(\alpha_{1}-\alpha_{2})\epsilon}+e^{(\alpha_{2}-\alpha_{1})\epsilon}-2)]
e^{-\frac{\alpha_{1}+\alpha_{2}}{x}+(\beta_{1}+\beta_{2})t+\eta_{1}+\eta_{2}}.
\end{equation}
Therefore, the form of $f^{(2)}$ can be
\begin{equation}\label{equa33}
f^{(2)}=A(1,2)e^{-\frac{\alpha_{1}+\alpha_{2}}{x}+(\beta_{1}+\beta_{2})t+\eta_{1}+\eta_{2}}.
\end{equation}
Substituting such $f^{(2)}$ into eq.\eqref{equa32} will help us determine the position of two-q-soliton as
\begin{equation}\label{equa34}
A(1,2)=-\frac{(\beta_{1}-\beta_{2})^{2}-(e^{(\alpha_{1}-\alpha_{2})\epsilon}+e^{(\alpha_{2}-\alpha_{1})\epsilon}-2)}
{(\beta_{1}+\beta_{2})^{2}-(e^{(\alpha_{1}+\alpha_{2})\epsilon}+e^{-(\alpha_{1}+\alpha_{2})\epsilon}-2)}=-\frac{P(p_{1}-p_{2})}{P(p_{1}+p_{2})}.
\end{equation}
Supposing $f^{(3)}=0$, by the use of the dispersion relation eq.\eqref{equa30} the coefficient of $\varepsilon^{3}$ vanishes trivially and so do the rest of $\varepsilon^{j}, j>3$. That means we have a good truncation up to $\varepsilon^{3}$ which leads to the two-q-soliton solution as
\begin{equation}\label{equa35}
f(x,t)=1+e^{-\frac{\alpha_{1}}{x}+\beta_{1}t+\eta_{1}}+e^{-\frac{\alpha_{2}}{x}+\beta_{2}t+\eta_{2}}
+A(1,2)e^{-\frac{\alpha_{1}+\alpha_{2}}{x}+(\beta_{1}+\beta_{2})t+\eta_{1}+\eta_{2}}.
\end{equation}
Therefore, we illustrate the collision of two-q-solitons as
Figure \ref{2}.
\begin{figure}[h!]
\centering
\raisebox{0.85in}{}\includegraphics[scale=0.45]{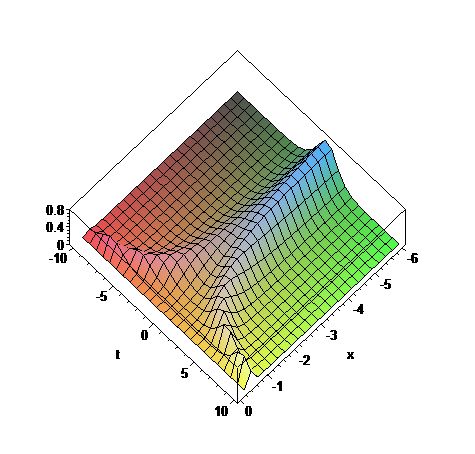}
\caption{Two-q-soliton solution $V$ of generalized $q$-Toda equation with $e^{\epsilon}=1.25,\alpha_{1}=-5,\alpha_{2}=6. $ }\label{2}
\end{figure}

To further derive three-soliton solution, we choose the starting solution of  eq.\eqref{equa23} as the assumption  as
\begin{equation}\label{equa36}
f^{(1)}=\sum\limits^{3}_{i=1}e^{-\frac{\alpha_{i}}{x}+\beta_{i}t+\eta_{i}},
\end{equation}
where $\alpha_{i}$,$\eta_{i}$ are arbitrary constants for $ i=1,2,3$. Similarly to the precious arguments, the coefficient of $\varepsilon^{0}$ vanishes
trivially. From the coefficient of $\varepsilon^{1}$, we have the corresponding dispersion relation
\begin{equation}\label{equa37}
\beta_{i}^{2}=e^{\alpha_{i}\epsilon}+e^{-\alpha_{i}\epsilon}-2, i=1,2,3.
\end{equation}
From the coefficient of $\varepsilon^{2}$, we can obtain
\begin{equation}\label{equa38}
-P(\partial)f^{(2)}=\sum\limits^{(3)}_{i<j}[(\beta _{i}-\beta_{j})^{2}-(e^{(\alpha_{i}-\alpha_{j})\epsilon}+e^{(\alpha_{i}-\alpha_{j})\epsilon}-2)]
e^{-\frac{\alpha_{i}+\alpha_{j}}{x}+(\beta_{i}+\beta_{j})t+\eta_{i}+\eta_{j}}.
\end{equation} The equation
eq.\eqref{equa38} implies the explicit form of $f^{(2)}$
\begin{equation}\label{equa39}
f^{(2)}=\sum\limits^{(3)}_{i<j}A(i,j)e^{-\frac{\alpha_{i}+\alpha_{j}}{x}+(\beta_{i}+\beta_{j})t+\eta_{i}+\eta_{j}},
\end{equation}
with
\begin{equation}\label{equa40}
A(i,j)=-\frac{P(p_{i}-p_{j})}{P(p_{i}+p_{j})}=-\frac{(\beta_{i}-\beta_{j})^{2}-(e^{(\alpha_{i}-\alpha_{j})\epsilon}+e^{(\alpha_{i}-\alpha_{j})\epsilon}-2)}
{(\beta_{i}+\beta_{j})^{2}-(e^{(\alpha_{i}+\alpha_{j})\epsilon}+e^{-(\alpha_{i}+\alpha_{j})\epsilon}-2)}.
\end{equation}
 For the coefficient of $\varepsilon^{3}$, we have
$$P(D)1\cdot f^{(3)}+f^{(3)}\cdot 1+f^{(1)}\cdot f^{(2)}+f^{(2)}\cdot f^{(1)}=0.$$
We can also rewrite them as
\begin{eqnarray}\notag
-P(\partial)f^{(3)}&=&(A(1,2)P(p_{3}-p_{1}-p_{2})+A(1,3)P(p_{2}-p_{1}-p_{3})+A(2,3)P(p_{1}-p_{2}-p_{3}))\\ \label{41}
&&\times
e^{-\frac{\alpha_{1}+\alpha_{2}+\alpha_{3}}{x}+(\beta_{1}+\beta_{2}+\beta_{3})t+\eta_{1}+\eta_{2}+\eta_{3}}.
\end{eqnarray}
Suppose that  $f^{(3)}$ is of the form
\begin{equation}\label{equa42}
f^{(3)}=A(1,2,3)e^{-\frac{\alpha_{1}+\alpha_{2}+\alpha_{3}}{x}+(\beta_{1}+\beta_{2}+\beta_{3})t+\eta_{1}+\eta_{2}+\eta_{3}},
\end{equation}
then one can find
\begin{equation}\label{equa43}
A(1,2,3)=-\frac{A(1,2)P(p_{3}-p_{1}-p_{2})+A(1,3)P(p_{2}-p_{1}-p_{3})+A(2,3)P(p_{1}-p_{2}-p_{3})}{P(p_{1}+p_{2}+p_{3})}.
\end{equation}
Following the steps, one can find we can suppose the vanishing of $f^{(4)}$ and it is a reasonable truncation to terms of  $\varepsilon^{4}$, i.e. the  from the equation eq.\eqref{equa22} becomes\\
\begin{equation}
2P(D)f^{(1)}\cdot f^{(3)}+P(Df^{(2)}\cdot f^{(2)}=0.
\end{equation}
which means the following condition holds
\begin{equation}\label{equa44}
A(1,2,3)=A(1,2)A(1,3)A(2,3).
\end{equation}
Then we can express the solution of three-q-soliton(see Figure \ref{3}) as
\begin{equation}\label{equa46}
\begin{split}
f(x,t)=1+\sum\limits^{3}_{i=1}e^{-\frac{\alpha_{i}}{x}+\beta_{i}t+\eta_{i}}
+\sum\limits^{3}_{i<j}A(i,j)e^{-\frac{\alpha_{i}+\alpha_{j}}{x}+(\beta_{i}+\beta_{j})t+\eta_{i}+\eta_{j}}\\
+A(1,2)A(1,3)A(2,3)e^{-\frac{\alpha_{1}+\alpha_{2}+\alpha_{3}}{x}+(\beta_{1}+\beta_{2}+\beta_{3})t+\eta_{1}+\eta_{2}+\eta_{3}}.
\end{split}
\end{equation}
\begin{figure}[h!]
\centering
\raisebox{0.85in}{}\includegraphics[scale=0.40]{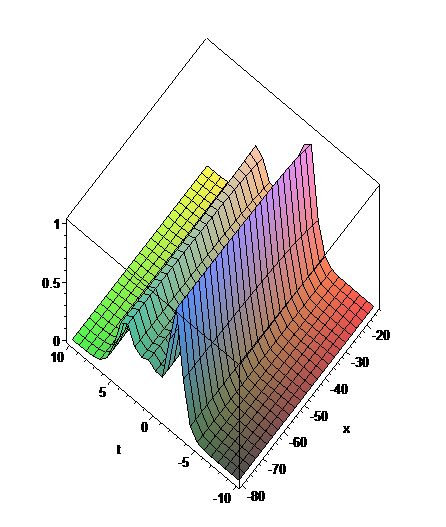}
\hskip 0.03cm
\raisebox{0.85in}{}\raisebox{-0.1cm}{\includegraphics[scale=0.50]{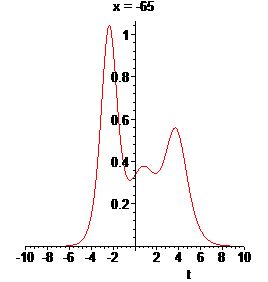}}
\caption{Three-q-soliton solution $V$ of generalized $q$-Toda equation with  $e^{\epsilon}=1.25,\alpha_{1}=-5,\alpha_{2}=6. \alpha_{3}=-7.9141,\beta_{1}=-1.1745,\beta_{2}=-1.4411, \beta_{3}=2.0045$} \label{3}
\end{figure}
\\
The above three-soliton solutions show  the great integrable possibility in a certain sense. To deeply prove the integrability, we will give the Lax pair of the
generalized  $q$-Toda\ hierarchy and further generalize it to a whole integrable hierarchy in the next section.

\section{The\  generalized $q$-Toda\ hierarchy}

Now we will consider that the algebra of the shift operator $\Lambda_{\epsilon} :=e^{\epsilon x^2\d_x}$. A Left
multiplication by  $X$ is as $X\Lambda_{\epsilon}^j$, $(
X\Lambda_{\epsilon}^j)(g)(x):=X(x)\circ g(\frac{x}{1-j\epsilon x})$ with defining the product
$(X(x)\Lambda_{\epsilon}^i)\circ(Y(x)\Lambda_{\epsilon}^j):=X(x)Y(\frac{x}{1-i\epsilon x})\Lambda_{\epsilon}^{i+j}.$

 Now we
introduce  the following free operators $ W_0,\bar  W_0$
\begin{align}
 \label{def:E}  W_0&:=e^{\sum_{j=0}^\infty
 t_{j}\frac{\Lambda_{\epsilon}^j}{\epsilon j!}}, \\
\label{def:barE}   \bar W_0&:=e^{-\sum_{j=0}^\infty
   t_{j}\frac{\Lambda_{\epsilon}^{-j}}{\epsilon j!}},
\end{align}
where $t_{j}\in \mathbb{R}$
will play the role of continuous times.

 We   define the dressing operators $W,\bar W$ as follows
\begin{align}
\label{def:baker}W&:=S\circ W_0,\ \  \bar W:=\bar S\circ \bar  W_0,
\end{align}
where  $S,\bar S$ have expansions as
\begin{gather}
\label{expansion-S}
\begin{aligned}
S&=1+\omega_1(x)\Lambda_{\epsilon}^{-1}+\omega_2(x)\Lambda_{\epsilon}^{-2}+\cdots,\\
\bar S&=\bar\omega_0(x)+\bar\omega_1(x)\Lambda_{\epsilon}+\bar\omega_2(x)\Lambda_{\epsilon}^{2}+\cdots.
\end{aligned}
\end{gather}

The inverse operators $S^{-1},\bar S^{-1}$ of operators $S,\bar S$ have expansions of the form
\begin{gather}
\begin{aligned}
S^{-1}&=1+\omega'_1(x)\Lambda_{\epsilon}^{-1}+\omega'_2(x)\Lambda_{\epsilon}^{-2}+\cdots,\\
\bar S^{-1}&=\bar\omega'_0(x)+\bar\omega'_1(x)\Lambda_{\epsilon}+\bar\omega'_2(x)\Lambda_{\epsilon}^{2}+\cdots.
\end{aligned}
\end{gather}
 The Lax  operator $\L$ of the generalized $q$-deformed Toda hierarchy
 is defined by
\begin{align}
\label{Lax}  \L&:=W\circ\Lambda_{\epsilon}\circ W^{-1}=\bar W\circ\Lambda_{\epsilon}^{-1}\circ \bar W^{-1},
\end{align}
and
have the following expansions
\begin{gather}\label{lax expansion}
\begin{aligned}
 \L&=\Lambda_{\epsilon}+U(x)+V(x)\Lambda_{\epsilon}^{-1}.
\end{aligned}
\end{gather}
 In fact the Lax  operators $\L$
 are also be equivalently defined by
\begin{align}
\label{two dressing}  \L&:=S\circ\Lambda_{\epsilon}\circ S^{-1}=\bar S\circ\Lambda_{\epsilon}^{-1}\circ \bar S^{-1}.
\end{align}

\subsection{ Lax equations of the GQTH}

In this section we will give the  Lax equations of the GQTH.
Let us firstly introduce some convenient notation such as  the operators $B_{j}$ defined as
$B_{j}:=\frac{\L^j}{j!}.$
Now we give the definition of the  generalized $q$-Toda hierarchy(GQTH).
\begin{definition}The  generalized $q$-Toda hierarchy is a hierarchy in which the dressing operators $S,\bar S$ satisfy following Sato equations
\begin{align}
\label{satoSt} \epsilon\partial_{t_{j}}S&=-(B_{j})_-S,& \epsilon\partial_{t_{j}}\bar S&=(B_{j})_+\bar S.\end{align}
\end{definition}
Then one can easily get the following proposition about $W,\bar W.$

\begin{proposition}The dressing operators $W,\bar W$ are subject to following Sato equations
\begin{align}
\label{Wjk} \epsilon\partial_{t_{j}}W&=(B_{j})_+ W,& \epsilon\partial_{t_{j}}\bar W&=-(B_{j})_-\bar W.  \end{align}
\end{proposition}

 From the previous proposition one can derive the following  Lax equations for the Lax operators.
\begin{proposition}\label{Lax}
 The  Lax equations of the GQTH are as follows
   \begin{align}
\label{laxtjk}
  \epsilon\partial_{t_{j}} \L&= [(B_{j})_+,\L].
  \end{align}
\end{proposition}

To see this kind of hierarchy more clearly, the  generalized $q$-Toda equations as the $t_{1}$ flow equations  will be given in the next subsection.
\subsection{The  generalized $q$-Toda equations}
 As a consequence Sato equations, after taking into account that   $S$ and $\bar S$, the $t_1$ flow of $\L$ in the form of $\L=\Lambda_{\epsilon}+U+V\Lambda_{\epsilon}^{-1}$ is as
\begin{gather}\label{exp-omega}
\begin{aligned}
  \epsilon\partial_{t_{1}} \L&= [\Lambda_{\epsilon}+U,V\Lambda_{\epsilon}^{-1}],
  \end{aligned}
\end{gather}
which lead to generalized $q$-Toda equation
\[\epsilon\partial_{t_{1}} U&=& V(\frac{x}{1-\epsilon x})-V(x),\\ \label{toda}
\epsilon\partial_{t_{1}} V&=& U(x)V(x)-V(x)U(\frac{x}{1+\epsilon x}).\]

From Sato equation we deduce the following set of nonlinear
partial differential-difference equations
\begin{align}\left\{
\begin{aligned}
 \omega_1(x)-\omega_1(\frac{x}{1+\epsilon x})&=\epsilon\partial_{t_1}(e^{\phi(x)})\cdot e^{-\phi(x)},\\
\epsilon\partial_{t_1}\omega_1(x)&=-e^{\phi(x)}e^{-\phi(\frac{x}{1-\epsilon x})}.\end{aligned}\right.
\label{eq:multitoda}
\end{align}
Observe that if we cross the two first equations, then we get
the generalized $q$-Toda equation \eqref{equaGQT}.
To give a linear description of the GQTH, we introduce  wave functions  $\psi,\bar\psi$
defined by
\begin{gather}\label{baker-fac}
\begin{aligned}
\psi&= W\cdot\chi, &
\bar\psi&=\bar W\cdot \bar\chi,
\end{aligned}
\end{gather}
where
\[
\chi(z):=z^{-\frac{1}{x \epsilon}},\ \ \bar \chi(z):=z^{\frac{1}{x\epsilon}},\
\]
and the $``\cdot"$ means the action of an operator on a function.
Note that $\Lambda_{\epsilon}\cdot\chi=z\chi$ and  the following asymptotic expansions
can be defined
\begin{gather}\label{baker-asymp}
\begin{aligned}
  \psi&=(1+\omega_1(x)z^{-1}+\cdots)\,\psi_0(z),&\psi_0&:=z^{-\frac{1}{x \epsilon}}
 e^{\sum_{j=1}^\infty t_{j}\frac{z^j}{\epsilon j!}},\\
\bar\psi&=(\bar\omega_0(x)+\bar\omega_1(x)z+\cdots)\,\bar\psi_0(z),
&\bar\psi_0&:=z^{\frac{1}{x \epsilon}}
e^{-\sum_{j=0}^\infty
   t_{j}\frac{z^{-j}}{\epsilon j!}}.
\end{aligned}
\end{gather}

We can further get linear equations of the GQTH in the following proposition.

\begin{proposition}The  wave functions $\psi,\bar\psi$ are subject to following Sato equations
\begin{align}
 \L\cdot\psi&=z\psi,\ \ \ &&\L\cdot\bar\psi=z\bar\psi,\\
 \epsilon\partial_{t_j}\psi&=(B_{j})_+\cdot \psi,& \epsilon\partial_{t_j}\bar \psi&=-(B_{j})_-\cdot\bar \psi.  \end{align}
\end{proposition}

\section{Bi-Hamiltonian structure and tau symmetry}

To describe the integrability of the GQTH, we will construct the Bi-Hamiltonian structure and tau symmetry of the GQTH in this section.
 In this section, we will consider the GQTH on Lax operator
 \[\L=\Lambda_{\epsilon}+u+e^v\Lambda_{\epsilon}^{-1}.\]
Then for $\bar f=\int  f dx, \bar g=\int g dx, $ we can define the hamiltonian bracket as
\[\{\bar f,\bar g\}=\int  \sum_{w,w'}\frac{\delta f}{\delta w}\{w,w'\}\frac{\delta g}{\delta w'} dx,\ \ w,w'=u\ or\ v.\]
The bi-Hamiltonian structure for the
GQTH can be given by the following two compatible Poisson brackets similar as \cite{CDZ,EQTH}

\begin{eqnarray}
&&\{v(x),v(y)\}_1=\{u(x),u(y)\}_1=0,\notag\\
&&\{u(x),v(y)\}_1=\frac{1}{\epsilon} \left[e^{\epsilon\,x^2\d_x}-1
\right]\delta(x-y),\label{toda-pb1}\\
&& \{u(x),u(y)\}_2={1\over\epsilon}\left[e^{\epsilon\,x^2\d_x}
e^{v(x)}-
e^{v(x)} e^{-\epsilon x^2\d_x}\right] \delta(x-y),\notag\\
&& \{ u(x), v(y)\}_2 = {1\over \epsilon}
u(x)\left[e^{\epsilon\,x^2\d_x}-1 \right]
\delta(x-y),\label{toda-pb2}\\
&& \{ v(x), v(y)\}_2 = {1\over \epsilon} \left[
e^{\epsilon\,x^2\d_x}-e^{-\epsilon x^2\d_x}\right]\delta(x-y).\notag
\end{eqnarray}
For any difference operator $A=
\sum_k A_k \Lambda_{\epsilon}^k$, define residue $Res A=A_0$.
In the following theorem, we will prove the above Poisson structure can be as the the Bi-Hamiltonian structures of the GQTH.
\begin{theorem}\label{GQTHbiha}
The flows of the GQTH  are Hamiltonian systems
of the form
\[
\frac{\d u}{\d t_{j}}&=&\{u,H_{j}\}_1, \ \ j\ge 0.
\label{td-ham}
\]

They satisfy the following bi-Hamiltonian recursion relation
\[\notag \{\cdot,H_{n-1}\}_2=n
\{\cdot,H_{n}\}_1.
\]
Here the Hamiltonians have the form
\begin{equation}
H_{j}=\int h_{j}(u,v; u_x,v_x; \dots; \epsilon) dx,\quad  \ j\ge 1,
\end{equation}
with
\[
 h_{j}&=&\frac1{j!} Res \, \L^{j}.
\]

\end{theorem}

\begin{proof}
The proof is similar as the proof in \cite{CDZ,EQTH}.
Here we will prove that the flows $\frac{\d}{\d t_{n}}$ are also
Hamiltonian systems with respect to the first Poisson bracket.

Suppose
\[
B_{n}=\sum_{k} a_{n;k}\, \Lambda_{\epsilon}^k,
\]
and from
\begin{equation}
  \label{edef3}
\frac{\partial \L}{\partial t_{ n}} = [ (B_{n})_+ ,\L ]= [ -(B_{n})_- ,\L ],
\end{equation}
we can derive equation
\[\epsilon\frac{\partial u}{\partial t_{ n}}&=&a_{n;1}(\frac{x}{1-\epsilon x})-a_{n;1}(x),\\
\epsilon\frac{\partial v}{\partial t_{ n}}&=&a_{n;0}(\frac{x}{1+\epsilon x}) e^{v(x)}-a_{n;0}(x) e^{v(\frac{x}{1-\epsilon x})}.
\]

By the following calcuation
\begin{eqnarray}
&&d \tilde h_{n}=\frac1{n!}\,d\, Res\left[\L^{n}
\right]= \frac1{n!}\, Res\left[\L^{n}
 d \L\right]\notag\\
&&= Res\left[a_{n;0}(x)du+a_{n;1}(\frac{x}{1+\epsilon x}) e^{v(x)}dv\right],
\end{eqnarray}
it yields the following identities
\begin{equation}\label{dH1-u12}
\frac{\delta H_{n}}{\delta u}=a_{n;0}(x),\quad \frac{\delta H_{n}}
{\delta v}=a_{n;1}(\frac{x}{1+\epsilon x}) e^{v(x)}.
\end{equation}
This agree with Lax equation

\[
\frac{\d u}{\d t_{n}}&=&\{u,H_{n}\}_1={1\over \epsilon} \left[
e^{\epsilon\,x^2\d_x}-1\right]\frac{\delta H_{n}}
{\delta v}={1\over \epsilon}(a_{n;1}(\frac{x}{1-\epsilon x})-a_{n;1}(x)),\\
 \  \frac{\d v}{\d t_{n}}&=&\{v,H_{n}\}_1=\frac{1}{\epsilon} \left[1-e^{\epsilon\,x^2\d_x}
\right]\frac{\delta H_{n}}
{\delta u}=\frac{1}{\epsilon} \left[a_{n;0}(\frac{x}{1+\epsilon x}) e^{v(x)}-a_{n;0}(x) e^{v(\frac{x}{1-\epsilon x})}\right].
\]

 From the above identities we see that
the flows $\frac{\d}{\d t_{n}}$ are Hamiltonian systems
with the first Hamiltonian structure.
The recursion relation
follows from the following trivial identities
\begin{eqnarray}
&&n\, \frac{1}{n!} \L^{n} =\L\,
\frac{1}{(n-1)!}
\L^{n-1}=\frac{1}{(n-1)!}
\L^{n-1}\L.\notag
\end{eqnarray}
Then we get,
\begin{eqnarray}
&&n a_{n;1}(x)=a_{n-1;0}(\frac{x}{1-\epsilon x})+ua_{n-1;1}(x)+e^va_{n-1;2}(\frac{x}{1+\epsilon x})\notag\\
&&=a_{n-1;0}(x)+u(\frac{x}{1-\epsilon x})a_{n-1;1}(x)+e^{v(\frac{x}{1-2\epsilon x})}a_{n-1;2}(x).\notag
\end{eqnarray}
This further leads to

\begin{eqnarray}
&&\{u,H_{n-1}\}_2=\{\left[\Lambda_{\epsilon} e^{v(x)}-e^{v(x)} \Lambda_{\epsilon}^{-1}\right] a_{n-1;0}(x)+
u(x) \left[\Lambda_{\epsilon}-1\right] a_{n-1;1}(\frac{x}{1+\epsilon x}) e^{v(x)}\}\notag\\ \notag
&&
=n\left[a_{n;1}(x) e^{v(\frac{x}{1-\epsilon x})}-a_{n;1}(\frac{x}{1+\epsilon x}) e^{v(x)}\right].\label{pre-recur}
\end{eqnarray}
This is exactly the recursion relation on flows for $u$. The similar recursion flow on $v$ can be similarly derived.
Theorem is proved till now.

\end{proof}

Similarly as \cite{CDZ}, the tau symmetry of the GQTH can be proved in the  following theorem.
\begin{theorem}\label{tausymmetry}
The GQTH has the following tau-symmetry property:
\begin{equation}
\frac{\d h_{m}}{\d t_{ n}}=\frac{\d
h_{ n}}{\d t_{m}},\quad \ m,n\ge 1.
\end{equation}
\end{theorem}
\begin{proof} Let us prove the theorem in a direct way
\[
&&\frac{\d h_{m}}{\d t_{n}} =\frac1{m!\,n!}\, Res[-(\L^{n})_-, \L^m]\notag\\
&&=\frac1{m!\,n!}\, Res[(\L^m )_+,(\L^{n})_-]\notag\\
&& =\frac1{m!\,n!}\, Res[(\L^m )_+,\L^{n}]=\frac{\d h_{n}}{\d t_{m}}.
\]
Theorem is proved.
\end{proof}

 This property justifies the definition of the
tau function for the GQTH as in the following proposition.

\begin{proposition} The $tau$ function of the GQTH can also be defined by
the following expressions in terms of the densities of the Hamiltonians:
\begin{equation}
h_{n}=\epsilon (\Lambda_{\epsilon}-1)\frac{\d\log  \tau}{\d t_{n}},
\quad \ n\ge 0.
\end{equation}
\end{proposition}

\section*{Acknowledgments}
Chuanzhong Li is supported by the National Natural Science Foundation of China under Grant No. 11201251,
  Zhejiang Provincial Natural Science Foundation of China under Grant No. LY12A01007, the Natural Science Foundation of Ningbo under Grant No. 2013A610105 and K.C.Wong Magna Fund in Ningbo University.

\end{document}